\documentclass[12pt,a4paper]{amsart}

\usepackage{amsmath,amsfonts,amssymb}

\usepackage[hmargin=2cm,vmargin=2cm]{geometry}

\usepackage[hyperfootnotes=false,colorlinks=true,linkcolor=blue,%
citecolor=purple,filecolor=magenta,urlcolor=cyan]{hyperref}%[hidelinks]{hyperref}
\usepackage{nameref,zref-xr}

\usepackage{amsthm,amscd}
\usepackage{epsfig}
\usepackage{color}
\usepackage[all]{xy}
\usepackage{tikz}
\usepackage{graphics, setspace}
\usepackage{breqn}
\usepackage{amstext}
\usepackage{array}

\newcommand{\CC}{{\mathbb{C}}}

\newcommand{\PP}{{\mathbb{P}}}
\newcommand{\QQ}{{\mathbb{Q}}}

\newcommand{\F}[2]{{F_{#1}^{#2}}}

\newcommand{\id}{\mathrm{id}}

\newcommand{\bx}{{\bf x}}

\def\A{{\mathcal A}}

\def\F{{\mathcal F}}

\def\R{{\mathcal R}}

\newcommand{\p}{{\partial}}

\newcommand{\ZZ}{\mathbb{Z}}

\newcommand{\bt}{{\bf t}}

\newtheorem{theorem}{Theorem}[section]
\newtheorem*{theorem*}{Theorem}
\newtheorem{proposition}[theorem]{Proposition}

\newtheorem*{corollary*}{Corollary}

\usepackage{color}

\def\&{\vspace{-5pt}&}

\makeatletter
\newsavebox{\@brx}
\newcommand{\llangle}[1][]{\savebox{\@brx}{\(\m@th{#1\langle}\)}%
  \mathopen{\copy\@brx\kern-0.5\wd\@brx\usebox{\@brx}}}
\newcommand{\rrangle}[1][]{\savebox{\@brx}{\(\m@th{#1\rangle}\)}%
  \mathclose{\copy\@brx\kern-0.5\wd\@brx\usebox{\@brx}}}
\makeatother

\allowdisplaybreaks

\numberwithin{equation}{section}

\begin{document}

\title{
BKP and CKP hierarchies via orbifold Saito theory
}

\author{Alexey Basalaev}
\address{A. Basalaev:
\newline
Faculty of Mathematics, HSE University, Usacheva str., 6, 119048 Moscow, Russian Federation
}
\email{abasalaev@hse.ru}

\date{\today}

 \begin{abstract}
 Semisimple Dubrovin–Frobenius manifolds can be used to construct integrable hierarchies, following the work of Dubrovin–Zhang and Buryak. Examples of such hierarchies include the Kac–Wakimoto hierarchies, the KP hierarchy, among others. In all these examples, the Saito theory of isolated singularities played a crucial role.

 In this note, we show that the BKP and CKP hierarchies can likewise be constructed from Dubrovin–Frobenius manifolds. This new construction, however, utilizes the orbifold version of Saito theory for isolated singularities endowed with a symmetry group.
 \end{abstract}
 \maketitle

%   \setcounter{tocdepth}{1}
%     \tableofcontents

\section{Introduction}

Tense connection between the cohomological field theories and integrable systems has been known since the last decade of XX. This was first conjectured by Witten (cf. \cite{W93}) that the partition functions of some cohomological field theories are tau--functions of the integrable hierarchies. Conjectures of Witten were proved by Kontsevich \cite{K92}, Faber-Shadrin-Zvonkine \cite{FSZ} and Fan-Jarvis-Ruan \cite{FJR}.

The genus $0$ potential of every cohomological field theory defines a formal Dubrovin--Frobenius manifold. Conversely if a Dubrovin--Frobenius manifold is semisimple, it can be used to reconstruct uniquely all genera potentials of some cohomological field theory uniquely.

It was later observed that Dubrovin--Frobenius manifolds can be used to construct integrable systems themselves rather then just tau-functions (cf. \cite{DVV}, \cite{D96}[Lecture 6] and other).

The general constructions of Dubrovin-Zhang \cite{DZ} and Buryak \cite{Bu} associate an integrable hierarchy to any semisimple Dubrovin--Frobenius manifold. Several special Dubrovin--Frobenius manifolds are known to have there specific constructions of the integrable hierarchies (cf. \cite{G03, FGM10, GM05, LRZ, P06, BDbN, B24, B25}). Some of these constructions were also extended to the flat F--manifolds (cf. \cite{ABLR, LPR, B22}).

Both constructions of Dubrovin--Zhang and Buryak applied to the ADE type Dubrovin--Frobenius manifolds give the hierarchy, equivalent to Kac--Wakimoto hierarchy of the same ADE type.
The full series of A--type Dubrovin--Frobenius manifolds can be used to construct the famous KP hierarchy.

It was essential to ask what Dubrovin--Frobenius manifolds provide BCFG Kac-Wakimoto hierarchies. This question was answered by Liu, Ruan and Zhang in \cite{LRZ}. In particular, they've found that BCFG Kac-Wakimoto hierarchies are given by the special subflows of the Dubrovin--Zhang hierarchies.
% of the certain Dubrovin--Frobenius manifolds.

Next to KP hierarchy well--known are BKP and CKP hierarchies (cf. \cite{DJKM}). Both are easy to define in the Lax form using the Lax operator of KP hierarchy and imposing just one additional constraint.
This is the second essential question to ask --- how to construct BKP and CKP hierarchies via the Dubrovin--Frobenius manifolds.

If such construction would exist, the corresponding Dubrovin--Frobenius manifolds would be isomorphic because BKP and CKP hierarchies are known to coincide dispersionless.

In this paper we make use of the orbifold Saito theories of A and D type singularities in order to construct BKP and CKP hierarchies.

\subsection*{Orbifold Saito theory}
ADE type Dubrovin--Frobenius manifolds have several equivalent definitions. One of them is via the Saito theory of the ADE type singularity (cf. \cite{S1, H}). In particular, one considers the unfolding of the polynomial $f \in \CC[\bx]$ defining the singularity in question. For $A_N$ and $D_N$ types the polynomials $f$ are given by $f_{A_N} := \frac
{1}{N+1}x_1^{N+1} + x_2^2$ and $f_{D_N} := \frac
{1}{N-1}x_1^{N-1} + x_1 x_2^2$.

Mirror symmetry sparkled investigation of the isolated singularities endowed with the symmetry groups. Following the physicists notation, the pairs $(f,G)$ where $f \in \CC[\bx]$ defines an isolated singularity and $G$ is its symmetry groups are now called \textit{Landau--Ginzburg orbifolds}.

Saito theory of the Landau--Ginzburg orbifolds has been under investigation of many researchers (cf. \cite{T1,T2,BT2,BR}). The corresponding Dubrovin--Frobenius manifold is $G$--graded. Its $\id$--graded submanifold is fixed by the $G$--invariants of the ``classical'' Saito theory and therefore called \textit{invariant sector}.

In this note we show that BKP and CKP hierarchies are constructed essentially via the Dubrovin--Frobenius manifolds of the Landau--Ginzburg orbifolds.
We consider the hierarchies of the orbifold Saito theory Dubrovin--Frobenius manifolds and show that restriction to the invariant sector gives a subhierarchy.

Our main theorem is the following.

\begin{theorem}\label{theorem: main}
    BKP hierarchy is the invariant sector hierarchy of the series of Landau--Ginzburg orbifolds $(A_{2N+1},\ZZ/2\ZZ)$.
    CKP hierarchy is the invariant sector hierarchy of the series of Landau--Ginzburg orbifolds $(D_{N},\ZZ/2\ZZ)$.
%     \begin{itemize}
%      \item BKP hiearchy is the invariant sector hierarchy of the series of Landau--Ginzburg orbifolds $(A_{2N-2},\ZZ/2\ZZ)$,
%      \item CKP hiearchy is the invariant sector hierarchy of the series of Landau--Ginzburg orbifolds $(D_{N+1},\ZZ/2\ZZ)$.
%     \end{itemize}
\end{theorem}
Proof of this theorem is given in Section~\ref{section: proof}.
\\
\\
\indent Compared to \cite{LRZ} we work with Saito theory Landau--Ginzburg orbifolds while the authors of loc.cit. work with the FJR theory. From the point of view of mirror symmetry we work on the B side while \cite{LRZ} work on the A side. Another difference is that we construct the BKP and CKP hierarchies rather then the Kac--Wakimoto hierarchies.
We also distinguish the subhierarchies inside the bigger hierarchies as \cite{LRZ} do. However in our case this subhierarchies are given essentially by the invariant sector whereas Liu, Ruan, Zhang do the specific research to sort out the suitable flows.

\subsection{Acknowledgements}
The work of Alexey Basalaev was supported by the Russian Science Foundation (grant No. 24-11-00366).

The author is grateful to A.V. Zabrodin for many fruitful discussions.

\section{Dubrovin--Frobenius manifolds}
Assume $M$ to be an open full--dimensional subspace of $\CC^l$. We say that it's endowed with a structure of Dubrovin--Frobenius manifold if there is a regular function $\F = \F(t_1,\dots,t_l)$ on $M$, s.t. the following conditions hold (cf. \cite{D96}).

\begin{itemize}
	\item
There is a distinguished variable $t_1$, such that:
\[
    \frac{\p \F}{\p t_1} = \frac{1}{2} \sum_{\alpha,\beta = 1}^l \eta_{\alpha,\beta} t_\alpha t_\beta,
\]
and $\eta_{\alpha,\beta}$ are components of a non-degenerate bilinear form $\eta$ (which does not depend on $t_\bullet$). In what follows denote by $\eta^{\alpha,\beta}$ the components of $\eta^{-1}$.

\item The function $\F$ satisfies a large system of PDEs called the WDVV equations:
\[
\sum_{\mu,\nu = 1}^l \frac{\p^3 \F}{\p t_\alpha \p t_\beta \p t_\mu} \eta^{\mu,\nu} \frac{\p^3 \F}{\p t_\nu \p t_\gamma \p t_\sigma}
=
\sum_{\mu,\nu = 1}^l \frac{\p^3 \F}{\p t_\alpha \p t_\gamma \p t_\mu} \eta^{\mu,\nu} \frac{\p^3 \F}{\p t_\nu \p t_\beta \p t_\sigma},
\]
which should hold for every given $1\leq\alpha,\beta,\gamma,\sigma\leq l$.

\item There is a vector field $E$ called the \textit{Euler vector field}, such that modulo quadratic terms in $t_\bullet$ we have $E \cdot \F = (3-\delta) \F$ for some fixed complex number $\delta$. We will assume $E$ to have the following simple form
\[
    E = \sum_{i=1}^l d_i t_i \frac{\p}{\p t_i}
\]
for some fixed numbers $d_1,\dots,d_l$.
\end{itemize}

Given such a data $(M,\F,E)$ one can endow every tangent space $T_pM$ with a structure of commutative associative product $\circ$ (depending on $\bt$) defined as follows:
\[
    \frac{\p}{\p t_\alpha} \circ \frac{\p}{\p t_\beta} = \sum_{\delta,\gamma=1}^l\frac{\p^3 \F}{\p t_\alpha \p t_\beta \p t_\delta} \eta^{\delta\gamma} \frac{\p}{\p t_\delta}.
\]
The unit of this product is the vector field $e = \frac{\p}{\p t_1}$.
It follows that $\eta(a \circ b,c) = \eta(a,b \circ c)$ for any vector fields $a,b,c$.

\subsection{A and D type Dubrovin--Frobenius manifolds}\label{sec:frobstruc}
Let $W$ stand either for $A_N$ of $D_N$.
The $A_N$ and $D_N$ type singularities are defined via the following polynomials in $x$ and $y$:
\[
    f_{A_N} = \frac{x^{N+1}}{N+1} +y^2 , \quad f_{D_N} = \frac{x^{N-1}}{N-1} + xy^2.
\]
One associates to them the so-called unfoldings $\Lambda_W: \CC^2\times\CC^N \to \CC$
\[
    \Lambda_{A_N} = \frac{x^{N+1}}{N+1} +y^2 + \sum_{k=1}^Nv_k x^{k-1}, \quad \Lambda_{D_N} = \frac{x^{N-1}}{N-1} + xy^2 + \sum_{k=1}^{N-1}v_k x^{k-1} + v_N y,
\]
that depend on additional parameters $v = (v_1,\dots,v_N)\in M_W := \CC^N$. %varying in $M_W := \CC^N$.

To introduce the Dubrovin--Frobenius manifold structure on $M_W$ consider for every fixed $v \in M_W$ the following quotient-ring:
\[
    \A_v := \CC[x,y]\, \Big/ \left( \frac{\p \Lambda_W}{\p x}, \frac{\p \Lambda_W}{\p y} \right).
\]
It's endowed with the quotient-ring product structure, and the classical singularity theory arguments assure that $\A_v$ is an $N$--dimensional $\CC$--vector space. Let $c_{ab}^s(v)$ stand for the structure constants of this product in the basis $[\p \Lambda_W/\p v_1],\dots,[\p \Lambda_W/\p v_N]$, namely,
\begin{align*}
    & A_N: \quad \frac{\p \Lambda_W}{\p v_k} = x^{k-1}, \ 1 \le k \le N,
    \\
    & D_N: \quad \frac{\p \Lambda_W}{\p v_k} = x^{k-1}, \ 1 \le k \le N-1, \quad \frac{\p \Lambda_W}{\p v_N} = y.
\end{align*}

The product $\circ: T_vM \otimes T_vM \to T_vM \otimes \CC[v_1,\dots,v_N]$ is now defined by
\[
    \frac{\p}{\p v_a} \circ \frac{\p}{\p v_b} := \sum_{k=1}^N c_{ab}^k(v) \frac{\p}{\p v_k}.
\]
Obviously, $\p/\p v_1$ is the unit of this product.
In particular, we have for $A_N$
\begin{align}
    \frac{\p}{\p v_a} \circ \frac{\p}{\p v_b} = \frac{\p}{\p v_{a+b-1}} \quad \forall a+b \le N+1,
    \label{eq: A_N simple str constants in v}
    \\
    \frac{\p}{\p v_a} \circ \frac{\p}{\p v_{N+2-a}} = - \sum_{k=2}^N (k-1)v_k \frac{\p}{\p v_{k-1}}.
\end{align}

Introduce the non-degenerate $\CC[v]$--bilinear pairing $\eta: T_vM \otimes T_vM \to \CC[v_1,\dots,v_N]$ by
\begin{align*}
    &\eta(\frac{\p}{\p v_a},\frac{\p}{\p v_b}) := c_{a,b}^N, \quad W = A_N,
    \\
    &\eta(\frac{\p}{\p v_a},\frac{\p}{\p v_b}) := c_{a,b}^{N-1}, \quad W = D_N.
\end{align*}
The pairing we introduce is in fact the well-known residue pairing.

\begin{theorem}[cf. \cite{D96,S1,ST}]
    The data $(M,\circ,\eta)$ is a Dubrovin--Frobenius manifold. In particular, there is a choice of the coordinates $t_\alpha = t_\alpha(v)$ (which are called the flat coordinates), s.t. in the basis $\p/\p t_1, \dots, \p/\p t_N$ we have
    \begin{itemize}
     \item the pairing $\eta$ is constant,
     \item there is a Frobenius potential $\F_W = \F_W(t_1,\dots,t_N)$, s.t.
     \[
        \frac{\p}{\p t_\alpha} \circ \frac{\p}{\p t_\beta} = \sum_{\gamma,\delta=1}^N \frac{\p^3 \F_W}{\p t_\alpha \p t_\beta \p t_\gamma} \eta^{\gamma \delta} \frac{\p}{\p t_\delta}.
     \]
    \end{itemize}
\end{theorem}

For the cases of $A_N$ and $D_N$ singularities, the flat coordinates of the theorem above were investigated by Noumi and Yamada in \cite{NY}.

\subsection{$A_N$ and $D_N$ Frobenius potentials}
% Let $W$ be either $A_N$ or $D_N$.
The potential $\F_W$ is a polynomial in $t_1,\dots,t_N$ with rational coefficients subject to the quasi-homogeneity condition $E_W \cdot \F_W = (3 - \delta_W) \F_W$ with
\begin{align*}
    & E_{A_N} = \sum_{\alpha=1}^{N}\frac{N+2-\alpha}{N+1} t_\alpha \frac{\p}{\p t_\alpha}, \quad & \delta_{A_N} = \frac{N-1}{N+1},
    \\
    & E_{D_N} = \sum_{\alpha=1}^{N-1}\frac{N-\alpha}{N-1} t_\alpha \frac{\p}{\p t_\alpha} + \frac{N}{2(N-1)} t_N \frac{\p}{\p t_N}, \quad &\delta_{D_N} = \frac{N-2}{N-1}.
\end{align*}
The pairing $\eta$ reads
\begin{align*}
    & \eta_{\alpha,\beta} = \delta_{\alpha+\beta,N+1} \quad & \text{ for } W = A_N,
    \\
    & \eta_{\alpha,\beta} =
    \begin{cases}
    1 \quad \text{when} \quad \alpha = \beta = N,
    \\
    \delta_{\alpha+\beta,N} \quad \text{otherwise}.
    \end{cases} \quad & \text{ for } W = D_N.
\end{align*}
We see that for all these $W$ for any $\alpha \in \{1,\dots,N\}$ there exists a unique integer $\bar\alpha \in \{1,\dots,N\}$ such that $\eta_{\alpha,\bar\alpha}= 1$.

For $W = A_N, D_N$ Noumi-Yamada gave the formulae for the potential $\F_W(t_1,\dots,t_N)$ in the following way. Consider the functions $\psi^{(r)}_\gamma \in \QQ[v_1,\dots,v_N]$ depending of the unfolding variables $v_k$ as above. Then for all $1 \le \alpha \le N$ we have:
\begin{equation}\label{eq: FW definition}
     \frac{\p \F_W}{\p t_\alpha} = \psi^{(2)}_{\overline \alpha}(t_1,\dots,t_N),
    \quad
    t_\alpha = \psi_\alpha^{(1)}(v_1,\dots,v_N).
\end{equation}

% It is only reasonable to consider the potential $F_W$ in flat coordinates $t_k$ and therefore it is important to invert the above formula of \cite{NY} in order to express $v_k = v_k(t)$.

\subsection{\texorpdfstring{$A_N$}{AN} case}
We have $\overline \alpha = N+1-\alpha$ and
\begin{align*}
    \psi_\gamma^{(r)}(v) &:= \sum_{\substack{\alpha_1,\dots,\alpha_N \ge 0 \\ \sum_{k=1}^{N} (N+2-k) \alpha_k = r(N+1) +1 - \gamma}} \left(-1\right)^{|\alpha|-r} \prod_{k=0}^{|\alpha|-1-r}(\gamma + k(N+1)) \prod _{k=1}^{N} \frac{v_k^{\alpha _k}}{\alpha _k!},
\end{align*}
where $|\alpha| = \sum_{k=1}^{N} \alpha_k$.

% The inverted formulae were given by Buryak in \cite{B1} from the study of open Gromov-Witten theories:
% \begin{align}\label{eq: An essential coordinate via flat}
%      v_\gamma &=  \sum_{\substack{\alpha_1,\dots,\alpha_N \ge 0 \\ \sum_{k=1}^{N} (N+2-k) \alpha_k = N+2 - \gamma}} \frac{(|\alpha|+\gamma-2)!}{(\gamma-1)!} \prod _{k=1}^{N} \frac{t_k^{\alpha _k}}{\alpha _k!}.
% \end{align}
% Note that the condition $\sum_{k=1}^{N} (N+2-k) \alpha_k = N+2 - \gamma$ precisely ensures that $v_\gamma$ is quasi-homogeneous (w.r.t. the Euler field $E_{A_N}$) and its weight is equal to the weight of $t_\gamma$.

\subsection{\texorpdfstring{$D_N$}{DN} case}
We have
\[
    \overline \alpha = N-\alpha, \ 1\le\alpha \le N-1, \quad \overline N = N.
\]
and
\begin{align*}
    \psi^{(1)}_{\gamma} &= \sum_{\substack{\alpha_1,\dots,\alpha_{N-1} \ge 0 \\ \sum_{k=1}^{N-1} (N-k)\alpha_k = N- \gamma}} \left(-1\right)^{|\alpha|-1} \prod _{k=0}^{|\alpha|-2} (2 \gamma -1 +2 k (N-1)) \prod_{k=1}^{N-1} \frac{v_k^{\alpha_k}}{\alpha_k!}, \quad 1 \le \gamma \le N-1,
    \\
    \psi^{(1)}_{N} &= v_N.
\end{align*}
where $|\alpha| = \sum_{k=1}^{N-1} \alpha_k$.

In order to introduce $\psi^{(2)}_\gamma$ consider the following combinatorial coefficients.
\begin{align*}
    A^{(1)}_{\gamma,\alpha} &:= \left(-1\right)^{|\alpha|-2} \prod _{k=0}^{|\alpha|-3} (2\gamma -1 + 2 k (N-1)),
    \\
    A^{(2)}_{\gamma,\alpha} &:= \left(-1\right)^{|\alpha|-1} \prod _{k=0}^{|\alpha|-2} (2 \gamma -1 + 2 k (N-1)), \qquad 1 \le \gamma \le N-2,
    \\
    A^{(2)}_{N-1,\alpha} &:= 2.
\end{align*}
Then for any $1 \le \gamma \le N-1$ we have
\begin{align*}
    \psi_\gamma^{(2)}(v) &:= \sum_{\substack{\alpha_1,\dots,\alpha_{N-1} \ge 0 \\ \sum_{k=1}^{N-1} (N-k)\alpha_k = 2(N-1) +1- \gamma}} A^{(1)}_{\gamma,\alpha} \prod _{k=1}^{n-1} \frac{v_k^{\alpha _k}}{\alpha _k!}
%      \\
%      &
    + \sum_{\substack{\alpha_1,\dots,\alpha_{N-1} \ge 0 \\ \sum_{k=1}^{N-1} (N-k)\alpha_k = N-1-\gamma}} \frac{A^{(2)}_{\gamma,\alpha}}{2} \prod _{k=1}^{n-1} \frac{v_k^{\alpha _k}}{\alpha _k!} \frac{v_N^2}{2},
    \\
    \psi_{N}^{(2)}(v) &:= v_1v_N,
\end{align*}

\subsection{Dubrovin--Frobenius submanifold}
This topic was investigated thoroughly by I.~Strachan in \cite{St}. Consider a Dubrovin--Frobenius manifold $(M, \circ, \eta)$. Let $\iota: M' \hookrightarrow M$ be a submanifold. Then $M'$ has an induced bilinear form $\eta'$ on it. Moreover, the pairing $\eta$ gives us the orthogonal projector $\mathrm{pr}: TM \to TM'$ that allows one to introduce the product $\circ'$ for $M'$
\[
    a \circ' b := \mathrm{pr} (a \circ b), \quad \forall a,b \in T_p M' \subset T_p M.
\]
It follows immediately that $(T_p M', \circ', \eta')$ satisfies the Frobenius algebra property for any $p \in M'$. Denote also $e':= \mathrm{pr}(e)$ and  $E':= \mathrm{pr}(E)$. Note that these vectors can even turn out to be zero.

We will call $M'$ a Dubrovin--Frobenius submanifold if $(M',\circ',\eta')$ is a Dubrovin--Frobenius manifold by itself. In particular, this assumes $\eta'$ to be non--degenerate and the product $\circ'$ to have a Frobenius potential (see \cite{St} for the details).

Assume also the product $\circ'$ above to be such that $a \circ' b = a \circ b$ in $T_pM$ for any $p \in M'$. This is equivalent to the property $((a \circ b) \mid_{M'} )^\perp = 0$. The Dubrovin--Frobenius submanifolds satisfying this property will be called \textit{natural}.

Assume Dubrovin--Frobenius submanifold $M'$ to be obtained by setting $t_k = 0$ for all $k \in \R$ with $\R \subset \lbrace 1,\dots,N \rbrace$. If $c_{\alpha\beta}^\gamma$ are structure constants of $\circ$, the property of being natural for $M'$ means
\[
    \frac{\p}{\p t_\alpha} \circ' \frac{\p}{\p t_\beta} = \sum_{\gamma \not\in \R} \left( c_{\alpha\beta}^\gamma \mid_{t_k = 0 \ \forall k \in \R} \right) \frac{\p}{\p t_\gamma}.
\]

The following are two major examples of the natural Dubrovin--Frobenius submanifolds.

\subsection{B type Dubrovin--Frobenius manifold}\label{section: BN}
It was shown by Zuber in \cite{Zu} that the polynomial $F_{B_N}$ defined by
\begin{equation}\label{eq: Bn via An}
    \F_{B_N}(t_1,\dots,t_N) := \F_{A_{2N-1}}(t_1,0,t_2,0,t_3,\dots,t_N)
\end{equation}
introduces the structure of a Dubrovin--Frobenius on $M_{B_N} := \CC^N$.
Its pairing satisfies $\eta_{\alpha,\beta} = \delta_{\alpha+\beta,N+1}$ and
\[
    E_{B_N} = \sum_{\alpha=1}^{N}\frac{N +1 -\alpha}{N} t_\alpha \frac{\p}{\p t_\alpha}, \quad \delta_{B_N} = \frac{N-1}{N}.
\]

One notes immediately that $M_{B_N}$ is a Dubrovin--Frobenius submanifold in $M_{A_{2N-1}}$ with $\iota (t_1,\dots,t_N) = (t_1,0,t_2,0,t_3,\dots,t_N)$. In the other words $M_{B_N}$ is given inside $M_{A_{2N-1}}$ by setting all even--indexed coordinates to zero.

It was also observed in \cite[Proposition~4.5]{BDbN} that the following equality holds.
\begin{equation}\label{eq: Bn via Dn}
    \F_{B_N}(t_1,\dots,t_N) = \F_{D_{N+1}}(t_1,t_2, \dots, t_{N},0).
\end{equation}

This gives us that $M_{B_N}$ is a Dubrovin--Frobenius submanifold in $M_{D_{N+1}}$ with $\iota (t_1,\dots,t_N) = (t_1,\dots,t_N,0)$. In the other words $M_{B_N}$ is given inside $M_{D_{N+1}}$ by setting $t_{N+1} = 0$.

\begin{proposition}\label{prop: natural submanifolds}
    $M_{B_N}$ is a natural Dubrovin--Frobenius submanifold in $M_{A_{2N-1}}$ and $M_{D_{N+1}}$.
\end{proposition}
\begin{proof}
    Consider $\iota: M_{B_N} \hookrightarrow M_{A_{2N-1}}$. Let $c_{\alpha\beta}^\gamma = c_{\alpha\beta}^\gamma(\bt)$ stand for the structure constants of $M_{A_{2N-1}}$. We need to show that the following property holds for all $1 \le a,b \le N$.
    \[
        \left( \frac{\p}{\p t_{2a-1}} \circ \frac{\p}{\p t_{2b-1}} \right) \mid_{t_{2k} = 0} \ = \ \sum_{d=1}^N c_{2a-1,2b-1}^{2d-1}  \mid_{t_{2k} = 0} \frac{\p}{\p t_{2d-1}}.
    \]
    Introduce the $\ZZ/2\ZZ$--action on $M_{A_{2N-1}}$ by
    \[
        h: (t_1,t_2,t_3,t_4,t_5,\dots,t_{2N-2},t_{2N-1}) \mapsto (t_1,-t_2,t_3,-t_4,t_5,\dots,-t_{2N-2},t_{2N-1}).
    \]
    It follows immediately from Eq.~\eqref{eq: FW definition} that $\F_{A_{2N-1}}$ is preserved under this action. Then
    \[
        h (c_{2a-1,2b-1}^k) = (-1)^{k-1} c_{2a-1,2b-1}^k.
    \]
    In particular, $c_{2a-1,2b-1}^{2d}$ are odd with respect to this $\ZZ/2\ZZ$--action. This means that it depends nontrivially on $t_{2k}$ and vanishes by $\iota$. This completes the first part.

    To prove that $\iota: M_{B_N} \hookrightarrow M_{D_{N+1}}$ is a natural Dubrovin--Frobenius manifold consider the $\ZZ/2\ZZ$--action on $M_{D_{N+1}}$
    \[
        h: (t_1,t_2,t_3, \dots,t_{N},t_{N+1}) \mapsto (t_1,t_2,t_3,\dots,t_{N},-t_{N+1}).
    \]
    It follows again from Eq.~\eqref{eq: FW definition} that $\F_{D_{N+1}}$ is preserved under this action and the proof copies the arguments above.
\end{proof}

The proof above makes use of some group action arguments. In fact, this group action is not just instrumental because it comes from the orbifold Saito theory ideas. The proposition above follows from the deeper results on the Gauss--Manin connection of orbifold Saito theory of A and D type singularities (cf. \cite{BR}).

%%%%%%%%%%%%%%%%%%%%%%%%%%%%%%%%%%%%%%%%%%%%%%%%%%%%%%%%%%%%%%%%%%%%%%%%%%%%%%%%%%%%%%%%%%%%%%%%%%%%%%%%
%%%%%%%%%%%%%%%%%%%%%%%%%%%%%%%%%%%%%%%%%%%%%%%%%%%%%%%%%%%%%%%%%%%%%%%%%%%%%%%%%%%%%%%%%%%%%%%%%%%%%%%%
%%%%%%%%%%%%%%%%%%%%%%%%%%%%%%%%%%%%%%%%%%%%%%%%%%%%%%%%%%%%%%%%%%%%%%%%%%%%%%%%%%%%%%%%%%%%%%%%%%%%%%%%
%%%%%%%%%%%%%%%%%%%%%%%%%%%%%%%%%%%%%%%%%%%%%%%%%%%%%%%%%%%%%%%%%%%%%%%%%%%%%%%%%%%%%%%%%%%%%%%%%%%%%%%%

\section{Integrable systems}

\subsection{Stabilizing potentials}
Let $\lbrace \F_N \rbrace_{N \ge N_{min}}$ be an infinite series of $N$--dimensional Frobenius potentials.
Assume $\F_N \in \CC[[t_1,\dots,t_N]]$. We will say that this series \textit{stabilize} if the following holds.

For any $N \ge N_{min}$, all indices $1 \le \alpha,\beta \le N$, satisfying the \textit{index stabilization condition} $\alpha + \beta \ll N$ and any $L > N$ the the following equality holds in $\CC[[s_1,\dots,s_L]]$.
\[
    \frac{\p^2 \F_N}{\p t_\alpha \p t_\beta} \mid_{ t_\alpha = \eta^{\alpha \gamma} s_\gamma}
    =
    \frac{\p^2 \F_L}{\p t_\alpha \p t_\beta} \mid_{ t_\alpha = \eta^{\alpha \gamma} s_\gamma}
\]
Note that the substitution $t_\alpha = \eta^{\alpha \gamma} s_\gamma$ on the both sides is taken with its own pairing $\eta$, different for the two sides.

The index stabilization condition can be made more precise in the particular examples.
It was proven in Theorem~4.1 of \cite{BDbN} that $\F_N = \F_{A_N}$ with $N_{min} = 1$ and index stabilization condition $\alpha+\beta \le N$ stabilize.

Another stabilizing series of potentials is given by the D type Dubrovin--Frobenius manifolds.
According to Theorem~4.9 of loc.cit. the series $\F_N = \F_{D_N}$ stabilize with $N_{min}=4$ and index stabilization condition $\alpha+\beta \le N-1$ if both $\alpha,\beta < N$ and $\alpha \le N-1$ if $\beta = N$.

\subsection{Dispersionless system}\label{section: int sys defintion}
To any stabilizing series $\lbrace \F_N \rbrace_{N \ge N_{min}}$ as above associate the system of PDEs on the function $f = f(t_1,t_2,t_3,\dots)$.

\begin{equation}\label{eq: main PDE}
    \p_\alpha \p_\beta f = \left( \frac{\p^2 \F_N}{\p t_\alpha \p t_\beta} \right) \mid_{t_\gamma = \eta^{\gamma \delta} \p_1 \p_\delta f},
\end{equation}
where we denote $\p_\alpha := \p / \p t_\alpha$ and the number $N$ of the right hand side is taken such that the index stabilization condition for $\alpha,\beta$ and $N$ holds.

Note that the right hand side of Eq.~\eqref{eq: main PDE} is a complex coefficients power series in $\p_1\p_\bullet f$.
This system of equations expresses the second order derivatives of $f$ with respect to the higher index variables via the set $\lbrace \p_1\p_kf \rbrace_{k=1}^\infty$ that should be understood as a {\it Cauchy data}.

The compatibility of this system of PDEs follows from WDVV equation on the potentials $\F_N$ (cf. Proposition~2.1 of \cite{BDbN} and proof of Theorem 1.2 in \cite{B24}).

It was shown in \cite{BDbN} that Eq.~\eqref{eq: main PDE} for the stabilizing system of A type Dubrovin--Frobenius manifolds is equivalent to dispersionless KP hierarchy. To be more precise Eq.~\eqref{eq: main PDE} coincides with the Fay form of dispersionless KP hierarchy after some rescaling.

For the stabilizing system of D type Dubrovin--Frobenius manifolds Eq.~\eqref{eq: main PDE}  is equivalent to dispersionless limit of 1--component reduced 2--component BKP hierarchy. To be more precise Eq.~\eqref{eq: main PDE} has two types of flows that originate from the specific structure of the potential. The first class of flows coincides with the Fay form of the dispersionless BKP hierarchy after some rescaling. The second class of flows coincides with the Fay form of the reduced dispersionless BKP hierarchy (cf. Section~7 of loc.cit.).
We will comment on this more in what follows.
% {\color{red}
\subsection{Dispersionfull systems}
There are several ways to associate to a semisimple Dubrovin--Frobenius manifold a dispersionfull hierarchy (cf. \cite{DZ,Bu}). In particular, the approach of Dubrovin and Zhang is to construct a dispersionless hierarchy first and extend it further to the dispersionfull one.

Our primary interest in this paper concerns A and D type Dubrovin--Frobenius manifolds.
In these special cases there are specific constructions of the dispersionfull hierarchies \cite{FGM10, GM05, G03, MT}.
% We will comment on them later on.
% }
\subsection{Dubrovin--Frobenius submanifolds and subflows}

Let $\lbrace \F_N \rbrace$ be a stabilizing series of potentials. Assume $\F_N$ to introduce the structure of Dubrovin--Frobenius manifold on $M_N$. Assume also $M'_N \hookrightarrow M_N$ to be a natural Dubrovin--Frobenius submanifold, given by
\[
    M_N' = \lbrace (t_1,\dots,t_N) \in M_N \ | \ t_\alpha = 0 \ \forall \alpha \in \R_N \rbrace
\]
for some index set $\R_N \subset \lbrace 1,\dots, N \rbrace$. Assume also $\R_{N} \subseteq \R_{N+1}$ for all $N \ge N_{min}$.

In particular, for $B_N \hookrightarrow A_{2N-1}$ we have $\R_{2N-1} = \lbrace 2,4,\dots, 2N-2 \rbrace$ and for $B_N \hookrightarrow D_{N+1}$ we have $\R_{N+1} = \lbrace N+1 \rbrace$.

\begin{proposition}\label{prop: subhierarchy}
 Let the series $\lbrace \F'_N \rbrace$ of Frobenius potentials stabilize. Then the system Eq.~\eqref{eq: main PDE} for it are subflows of the system Eq.~\eqref{eq: main PDE} for the series $\lbrace \F_N \rbrace$ obtained by taking $\p_1\p_k f = 0$ for all $k \in \cup_{N} \R_N$.
\end{proposition}
\begin{proof}
    It follows from the construction that for all $\alpha,\beta,\gamma \not\in \R_N$ we have
    \[
     \frac{\p^3 \F_N'}{\p t_\alpha \p t_\beta \p t_\gamma}
     =
     \frac{\p^3 \F_N}{\p t_\alpha \p t_\beta \p t_\gamma} \mid_{t_\delta = 0, \ \forall \delta \in \R_N}.
    \]
    Because $\F_N'$ is a Frobenius potential of a Dubrovin--Frobenius submanifold, for any $\alpha \in \R_N$ and any other index $\beta$, such that $\eta^{\alpha\beta} \neq 0$ we have $\beta \in \R_N$.
    Then Eq.~\eqref{eq: main PDE} for $\F'_N$ gives
    \begin{align*}
     \p_\alpha \p_\beta f
     & = \frac{\p^2 \F'_N}{\p t_\alpha \p t_\beta} \mid_{t_\nu = \eta^{\nu \delta} \p_1 \p_\delta f}
     = \frac{\p^2 \F_N}{\p t_\alpha \p t_\beta} \mid_{t_\delta = 0, \ \forall \delta \in \R_N} \mid_{t_\nu = \eta^{\nu \delta} \p_1 \p_\delta f}
     \\
     & = \left( \frac{\p^2 \F_N}{\p t_\alpha \p t_\beta} \mid_{t_\nu = \eta^{\nu \delta} \p_1 \p_\delta f} \right) \mid_{\p_1 \p_\delta f = 0, \ \forall \delta \in \R_N}.
    \end{align*}
\end{proof}

The proposition above shows an interesting connection between the natural Dubrovin--Frobenius submanifolds and subflows of the integrable systems. Recall that for the submanifolds we where considering the products $\p / \p t_\alpha \circ \p / \p t_\beta$ of the small number of the basis vectors - only those satisfying $\alpha,\beta \not\in \R_N$. On the PDEs side this translates to considering the flows $\p_\alpha\p_\beta f$ of an unknown function $f$ only with respect to the restricted set of variables.

Also, for the submanifolds, the coefficients of the product $\p / \p t_\alpha \circ \p / \p t_\beta$, when computed in the ambient Dubrovin--Frobenius structure, had to be restricted to $t_k = 0$ for $k \in \R_N$. On the PDEs side this translates to setting the respective Cauchy data to zero: $\p_1\p_k f = 0$.

\section{KP, BKP and CKP}
In this section we gather all the data we need about the KP, BKP and CKP hierarchies. We do not give full account of them referencing an interested reader to \cite{DJKM,NZ} and \cite{Za05,Za21}. We closely follow the work of Zabrodin below.

\subsection{KP hierarchy}
Consider the operators
\[
    D(z) := \sum_{n \ge 1} \frac{z^{-n}}{n} \p_{n}, \quad \Delta(z) := \frac{\exp(2\hbar D(z))-1}{\hbar}.
\]
Let $\bt = t_1,t_2,t_3,t_4,\dots$.
For a function $\tau = \tau(\bt)$ denote
\[
 \tau(\bt + [z^{-1}]) := \tau(t_1 + \frac{1}{z},t_2 + \frac{1}{2z^2},t_3 + \frac{1}{3 z^3},t_4 + \frac{1}{4 z^4},\dots) = \exp(D(z)) \cdot \tau(\bt).
\]

The function $\tau = \tau(\bt)$ is said to be a tau--function of KP hierarchy if it satisfies
\[
    \p_1 \log \frac{\tau(\bt + [z_1^{-1}]-[z_2^{-1}])}{\tau(\bt)}
    =
    (z_2-z_1) \left(
    \frac{\tau(\bt +[z_1^{-1}])\tau(\bt -[z_2^{-1}])}{\tau(\bt)\tau(\bt +[z_1^{-1}]-[z_2^{-1}])} - 1
    \right).
\]
KP hierarchy can be equivalently rewritten in a Fay form on the function $f = \hbar^2\log\tau$.
The Fay form of KP hierarchy is the following equation in the ring of formal power series in $z_1^{-1},z_2^{-1}$.
\[
    \exp(\Delta(z_1)\Delta(z_2) f ) = 1 - \frac{\Delta(z_1)\p_1f - \Delta(z_2)\p_1f}{z_1-z_2}.
\]
Assume also $f$ to have an $\hbar$--expansion $f = \sum_{g \ge 0} \hbar^g f_g$. The dispersionless limit of KP hierarchy in the Fay form is the following equation
\[
    \exp(D(z_1)D(z_2) f_0 ) = 1 - \frac{D(z_1)\p_1f_0 - D(z_2)\p_1f_0}{z_1-z_2}.
\]

\subsection{BKP hierarchy}\label{section: 1-BKP}
In what follows we need to distinguish between the even and odd--indexed variables. Denote  $\bt_o := t_1,t_3,t_5,\dots$.
% \[
%  \bt_o := t_1,t_3,t_5,\dots \ \text{and} \ \bt_e := t_2,t_4,t_6,\dots.
% \]

Consider the operators
\[
    D^{\mathrm{B}}(z) := \sum_{n \ge 0} \frac{z^{-2n-1}}{2n+1} \p_{2n+1}, \quad \Delta^{\mathrm{B}}(z) := \frac{\exp(2\hbar D^{\mathrm{B}}(z))-1}{\hbar}.
\]
For a function $\tau = \tau(\bt_o)$ denote
\[
 \tau(\bt_o + 2[z^{-1}]_o) := \tau(t_1 + \frac{2}{z},t_3 + \frac{2}{3 z^3},t_5 + \frac{2}{5 z^5},\dots) = \exp(2 D^{\mathrm{B}}(z)) \cdot \tau(\bt_o)
\]
and
\[
 \tau^{[z]_o} := \tau(\bt_o + 2[z^{-1}]_o), \quad \tau^{[z_1z_2]_o} := \tau(\bt_o + 2[z_1^{-1}]_o+ 2[z_2^{-1}]_o).
\]

The function $\tau = \tau(\bt_o)$ is said to be a tau--function of BKP hierarchy if it satisfies (see Section 3.2 of \cite{Za21})
\[
 \tau \tau^{[z_1z_2]_o} \left( 1 + \frac{1}{z_1 + z_2} \p_1 \log \frac{\tau}{\tau^{[z_1z_2]_o}} \right)
 =
 \tau^{[z_1]_o}\tau^{[z_2]_o} \left(1 + \frac{1}{z_1 - z_2} \p_1 \log \frac{\tau^{[z_1]_o}}{\tau^{[z_2]_o}}
 \right).
\]

BKP hierarchy can be rewritten in a Fay form on the function $f = \hbar^2\log\tau$.
\begin{align*}
    &\left(z_1 + z_2 - \p_1 \hbar \Delta^{\mathrm{B}}(z_1)\Delta^{\mathrm{B}}(z_2) f  - \p_1(\Delta^{\mathrm{B}}(z_1)f + \Delta^{\mathrm{B}}(z_2)f)\right)
    \exp(\Delta^{\mathrm{B}}(z_1)\Delta^{\mathrm{B}}(z_2) f)
%     \tag{BKP}\label{bkp}
     \\
     &\quad\quad\quad = \frac{z_1+z_2}{z_1-z_2} \left( z_1-z_2 - \p_1(\Delta^{\mathrm{B}}(z_1)f - \Delta^{\mathrm{B}}(z_2)f) \right).
     \notag
\end{align*}
Assume also $f$ to have an $\hbar$--expansion $f = \sum_{g \ge 0} \hbar^g f_g$. The dispersionless limit of BKP hierarchy in the Fay form is the following equation
\begin{align}
    \left( 1 - \frac{\p_1 (2D^{\mathrm{B}}(z_1) + 2D^{\mathrm{B}}(z_2))f_0}{z_1+z_2} \right) e^{2D^{\mathrm{B}}(z_1)\cdot2D^{\mathrm{B}}(z_2) f_0}
     = 1 - \frac{\p_1 (2D^{\mathrm{B}}(z_1) - 2D^{\mathrm{B}}(z_2))f_0}{z_1-z_2}.
\label{bkp-dl}
\end{align}

\subsection{CKP hierarchy}\label{section: CKP}
The function $\tau = \tau(\bt_o)$ is said to be a tau--function of CKP hierarchy if it satisfies (see Section 4.3 of \cite{Za21})
\begin{align}
    &
    \frac{z_1+z_2}{z_1-z_2} \Bigg[ \left( z_1 - \p_1 \log \frac{\tau^{[z_1z_2]_o}}{\tau^{[z_2]_o}} \right)^{1/2}
    \left( z_1 - \p_1 \log \frac{\tau^{[z_1]_o}}{\tau} \right)^{1/2}
    \notag
    \\
    &\quad - \left( z_2 - \p_1 \log \frac{\tau^{[z_1z_2]_o}}{\tau^{[z_1]_o}} \right)^{1/2}
    \left( z_2 - \p_1 \log \frac{\tau^{[z_2]_o}}{\tau} \right)^{1/2} \Bigg]
    \notag
    \\
    &\qquad = \left( z_1 + z_2 - \p_1 \log \frac{\tau^{[z_1z_2]_o}}{\tau} \right)
    \frac{\tau \tau^{[z_1z_2]_o}}{\tau^{[z_1]_o} \tau^{[z_2]_o}}.
\end{align}
The dispersionless limit of CKP hierarchy coincides with the dispersionless limit of BKP hierarchy --- Eq.~\eqref{bkp-dl}. This perfectly reflects the fact that the invariant sectors of $M_{A_{2N-1},\ZZ/ 2\ZZ}$ and $M_{D_{N+1},\ZZ / 2\ZZ}$ are isomorphic as the Dubrovin--Frobenius manifolds.

Equation on the tau--function above is more complicated than the respective equation for the BKP hierarchy. Working with the CKP hierarchy in what follows we use the following result of Krichever and Zabrodin. Essentially it shows that CKP hierarchy can be understood as the restriction of KP (see Section~2.3 of \cite{KZ}).

\begin{theorem}[Krichever, Zabrodin]\label{theorem: KZ on CKP}
    For any tau--function $\tau^{\mathrm{CKP}}$ of the CKP hierarchy there is a unique tau--function $\tau^{\mathrm{KP}}$ of the KP hierarchy, such that
    \[
        0 = \p_1\p_{2k} \tau^{\mathrm{KP}} \mid_{t_{2m} = 0, \ m \ge 1}, \quad \forall k \ge 1
    \]
    and
    \[
        \left( \tau^{\mathrm{CKP}}(t_1,t_3,t_5,\dots) \right)^2 = \tau^{\mathrm{KP}}(t_1,0,t_3,0,t_5,\dots).
    \]
\end{theorem}
This theorem is a part of Theorem~2.2 of \cite{KZ} together with some ingredients of its proof that are important in our exposition (Lemma~2.1 and Appendix~A of loc.cit.).

\subsection{A type Dubrovin--Frobenius manifold hierarchy}\label{sec:Ahier}
The system~\eqref{eq: main PDE} for the stabilizing series $\lbrace \F_{A_N} \rbrace$ was investigated in \cite{BDbN}.
It was proved in Theorem~6.2 of \cite{BDbN} that the system \eqref{eq: main PDE} coincides with the dispersionless limit of KP hierarchy after the substitution $t_k \mapsto t_{k}/k$.

Consider the differential operators $\p_k^\hbar$ defined by the equality
\[
    \Delta(z) = \sum_{k \ge 1} \frac{z^{-k}}{k} \p_k^\hbar.
\]
In particular, we have $\p_k^\hbar = \p_k + O(\hbar)$.
It follows from the results of Natanzon--Zabrodin \cite[Lemma~3.2]{NZ} that full KP hierarchy is obtained from the dispersionless one by the substitution $\p_k \mapsto \p_k^\hbar$. This gives us that the full KP hierarchy is obtained from the system~\eqref{eq: main PDE} for the stabilizing series $\lbrace \F_{A_N} \rbrace$ via this substitution.

\subsection{D type Dubrovin--Frobenius manifold hierarchy}\label{sec:Dhier}

The system~\eqref{eq: main PDE} for the stabilizing series $\lbrace \F_{D_N} \rbrace$ was investigated in \cite{BDbN}. The variable $t_N$ enters the potential $\F_{D_N}$ in a very special way.
As a consequence, this system of PDEs distributes in the following two
% Due to this it is useful to distribute the system of PDEs in two. It assumes the following form.
\begin{align}
 \p_\alpha\p_\beta f &= \frac{\p^2 \F_{D_N}}{\p t_\alpha \p t_\beta} \mid_{t_\gamma = \eta^{\gamma\delta} \p_1\p_\delta f}, \quad N \ge \alpha + \beta,
 \tag{dlD-1}\label{eq: dlD-1}
 \\
 \p_\alpha\p_0 f &= \frac{\p^2 \F_{D_N}}{\p t_\alpha \p t_N} \mid_{t_\gamma = \eta^{\gamma\delta} \p_1\p_\delta f}, \quad N \ge \alpha,
 \tag{dlD-2}\label{eq: dlD-2}
\end{align}
on the unknown function $f = f(t_0,t_1,t_2,\dots)$ with the Cauchy data $\p_1 \p_0 f, \p_1 \p_1 f, \p_1 \p_2 f, \dots$.

This is important to note that due to the special form of a D type Frobenius potential, the right hand side of Eq.~\eqref{eq: dlD-1} does not depend on $\p_1\p_0f$.
It was proved in Theorem~7.1 of \cite{BDbN} that the system \eqref{eq: dlD-1} above coincides with the dispersionless limit of BKP hierarchy after the substitution $t_k \mapsto t_{2k-1}/(2k-1)$ and $f \mapsto 2 f$. The flows \eqref{eq: dlD-2} where identified as well, however they do not play any role in the current paper.

\subsubsection{Dispersionfull hierarchy}\label{section: D type dispersionfull hierachy}
Similarly to the case of the stabilizing series $\lbrace F_{A_N} \rbrace $ and the KP hierarchy, we expect that BKP hierarchy can be obtained from its dispersionless limit via some $\hbar$--deformation. We are not aware about such a result to be established up to now.

According to \cite{LWZ} and \cite{FGM10,GM05} the hierarchy, associated to the D type series of Dubrovin--Frobenius manifolds is 2--component BKP hierarchy that is reduced in one of its components. The flows of Eq.~\eqref{eq: dlD-1} correspond to the full BKP component and are called positive in \cite{LWZ}. The flows of Eq.~\eqref{eq: dlD-2} correspond to the reduced BKP component and are called negative in loc.cit..

\section{Orbifold Saito theory}
In this section consider the polynomials $f_{A_{2N-1}}$ and $f_{D_N}$ introduced in Section~\ref{sec:frobstruc} together with two different symmetry groups $G$ both homeomorphic to $\ZZ/2\ZZ$.

For $f_{A_{2N-1}}$ let $G$ be generated by  $g: (x,y) \mapsto (-x,-y)$. And for and $f_{D_N}$ let $G$ be generated by  $g: (x,y) \mapsto (x,-y)$.
One notes immediately that such transformations preserve the polynomials given.

Mirror symmetry suggests existence of orbifold Saito theory associated with the pairs $(f_{A_{2N-1}},G)$ and $(f_{D_N},G)$. The attempts to construct an orbifold Saito theory were made in \cite{Tu1,Tu2, BT2, BR}.

Like the classical Saito theory, orbifold Saito theory should produce a Dubrovin--Frobenius manifold. Stemming from the orbifold theory this Dubrovin--Frobenius manifold should satisfy the certain special properties. In particular, its potential should depend on the variables $t_{g,\alpha}$ with $g \in G$ and $\alpha$ --- integer index.

The following theorem can be used as the definition of the Dubrovin--Frobenius manifolds of the Landau--Ginzburg orbifolds we are interested in.

\begin{theorem}
    Dubrovin--Frobenius manifold of the Landau--Ginzburg orbifold $(f_{A_{2N-1}},G)$ depends on the variables $t_{\id,1},\dots,t_{\id,N},t_{g,1}$ and is given by
\[
 \F_{(A_{2N-1},G)}(t_{\id,1},\dots,t_{\id,N},t_{g,1}) = \F_{D_{N+1}}(t_{\id,1},\dots,t_{\id,N},t_{g,1}).
\]

    Dubrovin--Frobenius manifold of the Landau--Ginzburg orbifold $(f_{D_{N+1}},G)$ depends on the variables $t_{\id,1},\dots,t_{\id,N},t_{g,1}$ and is given by
    \[
 \F_{(D_{N+1},G)}(t_{\id,1},\dots,t_{\id,N},t_{g,1},\dots,t_{g,N-1}) = \F_{A_{2N-1}}(t_{\id,1},t_{g,1},t_{\id,2},t_{g,2},\dots,t_{\id,N}).
\]
\end{theorem}

This theorem follows from \cite{Kau03,Kau06,BTW23,BT1} on the level of Frobenius algebras at the origin and from \cite{BR} on the level of Saito theories.

Main idea behind this theorem is that there is an orbifold equivalence $(A_{2N-1},G) \sim D_{N+1}$ and $(D_{N+1},G) \sim A_{2N-1}$ that can be obtained via the crepant resolution argument (see Section~5 of \cite{BTW23} and Section~6 of \cite{BT1}). Derived categories of the orbifold equivalent pairs are equivalent what suggests (but not proves) that the corresponding Dubrovin--Frobenius manifolds are isomorphic. Finally in \cite{BR} it was computed that the Saito theories match under this equivalence exaclty with the variable substitution as above.

% The potential of the Landau--Ginzburg orbifold $(f_{A_{2N-2}},G)$ is given by
% \[
%  F_{(A_{2N-2},G)}(t_{\id,1},\dots,t_{\id,N},t_{g,1}) = F_{D_{N+1}}(t_{\id,1},\dots,t_{\id,N},t_{g,1})
% \]
% The potential of the Landau--Ginzburg orbifold $(f_{D_{N+1}},G)$ is given by
% \[
%  F_{(D_{N+1},G)}(t_{\id,1},\dots,t_{\id,N},t_{g,1},\dots,t_{g,N-1}) = F_{A_{2N-2}}(t_{\id,1},t_{g,1},t_{\id,2},t_{g,2},\dots,t_{\id,N})
% \]

Denote by $M_{(f_{A_{2N-1}},G),\id}$ and $M_{(f_{D_{N+1}},G),\id}$ the submanifolds of $M_{(f_{A_{2N-1}},G)}$ and $M_{(f_{D_{N+1}},G)}$ respectively, obtained by setting all non--trivially graded variables to zero: $t_{g,\bullet} = 0$. We will call these submanifolds \textit{invariant sectors}.

General theory constitutes that invariant sectors are Dubrovin--Frobenius submanifolds. In the two cases we investigate in this paper this fact was introduced in Section~\ref{section: BN}. Proposition~\ref{prop: subhierarchy} shows existence of the subhierarchies associated to these submanifolds. In order to prove Theorem~\ref{theorem: main} it remains to show that these subhierarchies are exactly BKP and CKP respectively.

\subsection{Proof of Theorem~\ref{theorem: main}}\label{section: proof}
% \begin{proof}
%  \textbf{A type Landau--Ginzburg orbifold}.
 By the isomorphism $M_{(A_{2N-1},\ZZ/2\ZZ)} \cong M_{D_{N+1}}$ we have to investigate the system~\eqref{eq: main PDE} for the D type Dubrovin--Frobenius manifolds. Under this isomorphism $M_{(A_{2N-1},\ZZ/2\ZZ),\id}$ is identified with $M'_{D_{N+1}} := \lbrace t_{N+1}=0 \rbrace \subset M_{D_{N+1}}$.

%  \textbf{D type Landau--Ginzburg orbifold}.
 By the isomorphism $M_{(D_{N+1},\ZZ/2\ZZ)} \cong M_{A_{2N-1}}$ we have to investigate the system~\eqref{eq: main PDE} for the A type Dubrovin--Frobenius manifolds. Under this isomorphism $M_{(D_{N+1},\ZZ/2\ZZ),\id}$ is identified with $M'_{A_{2N-1}} := \lbrace t_{2}=\dots = t_{2N-2} = 0 \rbrace \subset M_{A_{2N-1}}$.

 By Proposition~\ref{prop: natural submanifolds}, both $M'_{D_{N+1}}$ and $M'_{A_{2N-1}}$ are natural submanifolds and by Proposition~\ref{prop: subhierarchy} their flows are given exactly by Eq.~\eqref{eq: dlD-1}. Theorem~7.1 of \cite{BDbN} completes the proof for the dispersionless case.

 Dispersionfull statement for A type Landau--Ginzburg orbifolds follows immediately because in the terminology of Section~\ref{section: D type dispersionfull hierachy} restricting to the invariant sector we consider only the positive flows on BKP hierarchy.

 Dispersionfull statement for D type Landau--Ginzburg orbifolds and CKP hierarchy follows from Theorem~\ref{theorem: KZ on CKP}.

% \end{proof}

%%%%%%%%%%%%%%%%%%%%%%%%%%%%%%%%%%%%%%%%%%%%%%%%%%%%%%%%%%%%%%%%%%%%%%%%%%%%%%%%%%%%%%%%%%%%%%%%%%%%%%%%
%%%%%%%%%%%%%%%%%%%%%%%%%%%%%%%%%%%%%%%%%%%%%%%%%%%%%%%%%%%%%%%%%%%%%%%%%%%%%%%%%%%%%%%%%%%%%%%%%%%%%%%%
%%%%%%%%%%%%%%%%%%%%%%%%%%%%%%%%%%%%%%%%%%%%%%%%%%%%%%%%%%%%%%%%%%%%%%%%%%%%%%%%%%%%%%%%%%%%%%%%%%%%%%%%
%%%%%%%%%%%%%%%%%%%%%%%%%%%%%%%%%%%%%%%%%%%%%%%%%%%%%%%%%%%%%%%%%%%%%%%%%%%%%%%%%%%%%%%%%%%%%%%%%%%%%%%%

\end{document}